\documentclass[%
 reprint,
 amsmath,amssymb,
 aps,
]{revtex4-2}
\usepackage{graphicx}
\usepackage{indentfirst}
\usepackage{braket}
\usepackage{float}
\usepackage{amsmath}
\usepackage{amssymb}
\usepackage{amsfonts}
\usepackage{amsthm}
\usepackage{epstopdf}
\usepackage{CJK}
\usepackage{esint}
\usepackage{color}
\usepackage[T1]{fontenc}
\usepackage{subfigure}
\usepackage{amsfonts}
\usepackage{natbib}
\usepackage{footmisc}
\usepackage{verbatim}
\usepackage{multirow}
\usepackage[english]{babel}
\usepackage{url}
\usepackage{bm}
\usepackage{hyperref}
\definecolor{darkblue}{rgb}{0,0,0.5}
\hypersetup{
colorlinks=true,
linkcolor=black,
filecolor=blue,
citecolor=darkblue,  
urlcolor=black,
}

\usepackage{enumitem,kantlipsum}
\usepackage{refcount}
\usepackage{svg}

\newcounter{lemmacounter}
\newenvironment{lemma}[1][]{\refstepcounter{lemmacounter}
   {{\em Lemma~\thelemmacounter #1}.---} \rmfamily}
   
\newcounter{theoremcounter}
\newenvironment{theorem}[1][]{\refstepcounter{theoremcounter}
   {{\em Theorem~\thetheoremcounter #1}.---} \rmfamily}
   
\newcounter{propositioncounter}
\newenvironment{proposition}[1][]{\refstepcounter{propositioncounter}
   {{\em Proposition~\thepropositioncounter #1}.---} \rmfamily}
   
\newcounter{conjecturecounter}
\newcounter{definitioncounter}
\newenvironment{definition}[1][]{\refstepcounter{definitioncounter}
   {{\em Definition~\thedefinitioncounter #1}.---} \rmfamily}

\def\be{\begin{equation}}
\def\ee{\end{equation}}
\def\ba{\begin{eqnarray}}
\def\ea{\end{eqnarray}}
\newcommand{\calI}{{\cal I}}

\newcommand{\tr}{{\rm Tr}}

\newcommand{\eq}[1]{(\hyperref[eq:#1]{\ref*{eq:#1}})}

\renewcommand{\sec}[1]{\hyperref[sec:#1]{Section~\ref*{sec:#1}}}
\newcommand{\thrm}[1]{\hyperref[thm:#1]{Theorem~\ref*{thm:#1}}}
\newcommand{\lemm}[1]{\hyperref[lemm:#1]{Lemma~\ref*{lemm:#1}}}
\newcommand{\pro}[1]{\hyperref[pro:#1]{Proposition~\ref*{pro:#1}}}
\newcommand{\corr}[1]{\hyperref[corr:#1]{Corollary~\ref*{corr:#1}}}
\newcommand{\deff}[1]{\hyperref[deff:#1]{Definition~\ref*{deff:#1}}}
\newcommand{\fig}[1]{\hyperref[fig:#1]{\ref*{fig:#1}}}
\newcommand{\tbl}[1]{\hyperref[fig:#1]{\ref*{tbl:#1}}}

\DeclareMathOperator{\Tr}{Tr}

\usepackage{graphicx}
\usepackage{dcolumn}
\usepackage{bm}

\begin{document}

\preprint{APS/123-QED}

\title{No-broadcasting of non-Gaussian states}

\author{Kaustav Chatterjee}
 \email{kauch@dtu.dk}
\author{Ulrik Lund Andersen}%
\affiliation{%
Center for Macroscopic Quantum States (bigQ), Department of Physics, Technical University of Denmark, \\
 Building 307, Fysikvej, 2800 Kongens Lyngby, Denmark 
}%


\date{\today}

\begin{abstract}
Gaussian states are of fundamental importance in the physics of continuous-variable quantum systems.
They are appealing for the experimental ease with which they can be produced, and for their compact
and an elegant mathematical description. Nevertheless, many proposed quantum technologies require us to go beyond the realm of Gaussian states and introduce non-Gaussian elements. In terms of quantum resource theory, we can then recognize non-Gaussian states as resources and Gaussian operations and states as free, which can be used and prepared easily. Given such a structure of resource theory, the task of broadcasting the resource is to determine if the resource content of a state can be cloned in a meaningful way, which, if possible, provides a strong operation for manipulation of the resource. In this work, we prove that broadcasting of non-Gaussian states via Gaussian operations is not possible. For this, we first show that the relative entropy of non-Gaussianity is not super-additive, which rules it out as a prime candidate in the analysis of such no-go results. Our proof is then based on understanding fixed points of Gaussian operations and relates to the theory of control systems. The no-go theorem also states that if two initially uncorrelated systems interact by Gaussian dynamics and non-Gaussianity is created at one subsystem, then the non-Gaussianity of the other subsystem must be reduced. Further, keeping the set of free operations fixed to Gaussian operations, we can also comment on the broadcasting of Wigner negativity and genuine quantum non-Gaussianity.   
\end{abstract}

\maketitle


\section{\label{sec:level1}Introduction\protect}
Continuous-variable (CV) quantum information \cite{asera,andersen2010continuousvariablequantuminformation,Braunstein_2005} leverages the infinite-dimensional Hilbert space of bosonic modes, like the quantized electromagnetic field or vibrational modes of a mechanical oscillator. Rooted in quantum optics, these systems offer a powerful and complementary framework for quantum information processing alongside traditional discrete-variable approaches.

Gaussian states and operations are central to continuous-variable quantum information processing. Even though these states exist in an infinite-dimensional Hilbert space, their Gaussian nature of the characteristic functions allows us to derive analytical results with relative ease. Moreover, their experimental accessibility makes them a practical choice for implementing vital quantum protocols like teleportation \cite{Pirandola_2006}, enhanced sensing \cite{Pirandola_2018}, and key distribution \cite{Zhang_2024}. Unfortunately, such Gaussian schemes are limited in their power of CV quantum information processing. It has been shown that non-Gaussianity in the form of either non-Gaussian states or non-Gaussian operations is required for entanglement distillation \cite{ED1,ED2}, error correction \cite{ec1}, loophole-free violation of Bell’s inequality \cite{be1}, and universal quantum computation \cite{Qc1,Qc2}. These considerations elevate non-Gaussianity to a resource that can be quantitatively accounted for using the framework of quantum resource theory \cite{Chitambar_2019}.

A fundamental question in any resource theory is whether the resource can be \emph{broadcast} or cloned using only free operations. The concept of broadcasting generalizes cloning by demanding that a single resource state be distributed to two or more parties such that each party ends up holding a share of that resource.
No-broadcasting theorems are known for various quantum resources. For example, it has been established that the total correlations in a bipartite state can be broadcast if and only if the state is classical-classical, while quantum correlations can only be broadcast if the state exhibits a classical-quantum structure \cite{EB1, EB3,EB4,EB5}. Moreover, although the entanglement of bipartite states may be partially broadcast via local operations, exact broadcasting is not possible \cite{EB2}. Furthermore, neither the coherence nor the asymmetry of quantum states can be broadcast \cite{Asy1,asy2}, and broadcasting any magic state through stabilizer operations is ruled out \cite{magic}. But broadcasting imaginarity is again possible \cite{imagine}. Finally, the broadcasting of thermodynamic athermality is dependent on the bath temperature: it is impossible at any positive temperature but becomes feasible at absolute zero \cite{son2024robustcatalysisresourcebroadcasting}. These examples indicate that the task of broadcasting varies significantly based on the resource under consideration. In this work, we analyse the broadcasting problem for non-Gaussianity resource in CV quantum systems. We prove that non-Gaussian quantum states cannot be broadcast by Gaussian operations. 

Conceptually, our no-broadcasting theorem reveals a foundational limitation on Gaussian dynamics: these operations are so “resource-non-generative” that they cannot even distribute existing non-Gaussianity between two modes. This result is analogous to the no-broadcasting of quantum correlations, but cast in terms of a resource (non-Gaussianity) and a restricted class of operations (Gaussian channels). The core insight behind the proof is that Gaussian channels are \emph{strictly contractive} on phase-space distributions, meaning they inevitably drive states towards the Gaussian manifold. 
We also discuss this no-go theorem through the lens of resource degradability, as done for asymmetry \cite{Asy1}. We further discuss what implications it has on resources like Wigner negativity and genuine non-Gaussianity \cite{Walschaers_2021}. The remainder of the work is organized as follows. In Sec. II, we review the CV quantum information formalism and the corresponding resource theory of non-Gaussianity. In Sec. III, we discuss the task of broadcasting and properties of resource measures that relate to the task. Then we derive some no-broadcasting results for non-Gaussian states by analyzing the fixed points of Gaussian operations. In Sec. IV, we conclude with a discussion and summary. In the Appendices, we present detailed mathematical proofs of the main results.    

\section{\label{sec:level2} PRELIMINARIES\protect}
\subsection{\label{sec:level2}Gaussian states and Gaussian operations}
An $N$-mode bosonic continuous-variable system is described by annihilation operators $\left\{\hat{a}_k, 1\le k \le N\right\}$, which satisfy the commutation relation $\left[\hat{a}_k,\hat{a}_j^\dagger\right]=\delta_{kj}, \left[\hat{a}_k,\hat{a}_j\right]=0$. Equivalently, one can define 2N real quadrature field operators $\hat{q}_k=\frac{1}{\sqrt{2}}(\hat{a}_k+\hat{a}_k^\dagger), \hat{p}_k=\frac{i}{\sqrt{2}}\left(\hat{a}_k^\dagger-\hat{a}_k\right)$ and collect them into the real vector
$\hat{x}=\left(\hat{q}_1,\hat{q}_2,\cdots,\hat q_N, \hat{p}_1,\hat{p}_2,\cdots,\hat p_N\right)$. This vector satisfies the canonical commutation relation
$
\left[\hat{x}_i,\hat{x}_j\right]=i {  \Omega}_{ij}.
$
where $\Omega=\begin{pmatrix}
    0 & I_n\\
    -I_n & 0\\
\end{pmatrix}$. A quantum state $\rho$ can be conveniently described by its (symmetrically ordered) characteristic function
\be 
\chi\left({  \xi};\rho\right)=\tr \left[\rho \hat{D}\left({  \xi}\right)\right],
\label{Wigner_characterisitic_function}
\ee
where 
$
\hat{D}\left({  \xi}\right)=\exp\left(i\hat x{  \xi}\right)$
is the multi-mode Weyl displacement operator and $  \xi=\left(\xi_1,\cdots \xi_{2N}\right)^T\in \mathbb{R}^{2N}$ is a phase-space vector. A state $\rho$ is Gaussian if and only if its characteristic function has the Gaussian form \cite{asera}
\be
\chi \left({  \xi};\rho\right)=\exp\left(-\frac{1}{4}{  \xi}^T \left( {  \Lambda }\right){  \xi}+i {  \overline{x}}^T{  \xi}\right).
\ee
Here $\overline{  x}^T=\braket{  \hat{x}}_{\rho}$ is the state's mean vector and 
$
{  \Lambda}_{ij}=\braket{\{\hat{x}_i-\overline{x}_i,\hat{x}_j-\overline{x}_j\}}_\rho$
is its covariance matrix, with $\{,\}$ denoting the anticommutator. 
Thus, every Gaussian state is completely characterized by $\overline{  x}$ and $\Lambda$. 

The class of completely positive trace-preserving (CPTP) maps (quantum channels) that transform Gaussian states to Gaussian states are called Gaussian quantum operations. Denoting the set of Gaussian operations by $\mathcal{O}$ and the set of Gaussian states by $\mathcal{F}_G$, one has $\mathcal{O}(\mathcal{F}_G)=\mathcal{F}_G$ \cite{asera,Caruso_2008,eisert2005gaussianquantumchannels}. Any quantum channel $\mathcal{E}\in\mathcal{O}$ admits a Stinespring dilation of the form $\mathcal{E}(\rho)=\tr_E(\hat{U}_{  S,  d}(\rho\otimes\gamma_E)\hat{U}_{  S,  d}^\dagger)$, where $\hat{U}_{  S,  d}$ is a Gaussian unitary operation and $\gamma$ is a Gaussian state of the environment. Without loss of generality, the environmental state can be taken to have the same number of modes as the system.
A Gaussian unitary $\hat{U}_{  S,  d}$ transforms
\be
\hat{U}_{  S,  d}^\dagger \hat{  x} \hat{U}_{  S,  d}= \hat{  x}S+{  d},
\ee
where ${  d}=\left(d_1,\cdots, d_{2N}\right)$ is a displacement vector and $S$ is a symplectic matrix satisfying $S\Omega S^T=\Omega$. This form is different from the standard forms because the operator $\hat x$ in our case is a row vector. Consequently, a Gaussian unitary corresponds to a linear coordinate transform on the Wigner characteristic function,
\ba\label{t3}
\chi \left({  \xi};\hat{U}_{  S,  d} \rho \hat{U}_{  S,  d}^\dagger\right)&=&\chi \left({  S}{  \xi};\rho\right) \exp\left(i { d^T} {  \xi }\right).
\ea
Most of the properties and proofs that we discuss are independent of the displacement $d$, and we will omit it unless stated otherwise. Given a Gaussian operation with dilation $\hat{U}_{S}$ and environment $\gamma$, its action on the covariance matrix is 
\be\label{x2}
\Lambda\to X^T\Lambda X +Y
\ee
where $X=s_1^T$ and $Y=s_2\Lambda_\gamma s_2^T$ where $s_1,s_2$ are matrices that come from a global symplectic $S=\begin{pmatrix} s_1 & s_2\\
s_3&s_4\end{pmatrix}$ and $\Lambda_\gamma$ is the covariance matrix of the environment state. This can be seen from the Stinespring dilation of the channel, which maps $\rho\to \tr_E(\hat{U}_S\rho\otimes\gamma\hat{U}^\dagger_S)$, which on the level of the covariance matrix first maps as $\begin{pmatrix} s_1 & s_2\\
s_3&s_4\end{pmatrix}\begin{pmatrix} \Lambda & 0\\
0&\Lambda_\gamma\end{pmatrix}\begin{pmatrix} s_1^T & s_3^T\\
s_2^T&s_4^T\end{pmatrix}$ and the tracing out projects into the system part to give $\Lambda\to s_1\Lambda s_1^T+s_2\Lambda_\gamma s_2^T$. Notice that as $\Lambda_\gamma>0$, $Y\geq 0$ with $Y=0$ if and only if $s_2=0$. Also, throughout the paper, our conventions for phase space vectors and multi-partite splittings would be adopted to \cite{Caruso_2008}.  

Since we will frequently use the Wigner characteristic function, we summarize three important properties:
\begin{enumerate}
    \item Transformation under Gaussians \cite{asera}:
    \be\label{chip1x}
    \chi \left({  \xi};\mathcal{E}(\rho)\right)=\exp \left( -\frac{1}{4}\xi^T Y \xi\right)\chi \left( X{  \xi};\rho\right)
    \ee
    \item Tensor product states:
    \be\label{chip2}
    \chi\left({  \xi}_A,{  \xi}_B;\rho_A\otimes\rho_B\right)=\chi\left({  \xi}_A;\rho_A\right)\chi\left({  \xi}_B;\rho_B\right)
    \ee
    \item Partial trace: For a bipartite state $\rho_{AB}$ with characteristic function $\chi\left({  \xi}_A,{  \xi}_B;\rho_{AB}\right)$,
    \be\label{chip3}
    \chi\left({  \xi}_B;\rho_B\right)=\chi\left(0,{  \xi}_B;\rho_{AB}\right)
    \ee
\end{enumerate}
Proofs of properties $2$ and $3$ can be found in the appendix \ref{appen1}. Finally, note that because of state normalization, we always have $\chi\left(0;\rho\right)=1$, which follows directly from the definition (\ref{Wigner_characterisitic_function}).

\subsection{\label{sec:level2}Resource theory framework} 
Resource theories provide a systematic framework for quantifying and managing specific valuable quantities within a given context. In these frameworks, states lacking the resource are termed free states, while operations that cannot generate the resource from free states are known as free operations. Notably, free operations leave the set of free states unchanged. Given a set of free states, what operations are deemed free may vary, and it is possible to define different classes of free operations corresponding to the same class of free states. A prime example is the resource theory of quantum coherence \cite{qcoh}. In our setting, we primarily consider $\mathcal{F}_G$ as the set of free states and $\mathcal{O}$ as the set of free operations. It is important to note that attaching arbitrary Gaussian states (using tensor product to increase system size) and tracing out an arbitrary number of modes are both included in $\mathcal{O}$. This resource theory is not convex in the sense that a convex combination of free states can itself be resourceful. In this regard, one can define two additional notions of free states:
\begin{enumerate}
    \item $\mathcal{F}^c_G$, the convex hull of $\mathcal{F}_G$.
    \item $\mathcal{F}_W$, the set of states with non-negative Wigner functions, where the Wigner function is defined as the Fourier transform of the Wigner characteristic function.
    
\end{enumerate}
It is interesting to note that $\mathcal{F}_G\subsetneq \mathcal{F}^c_G\subsetneq \mathcal{F}_W$ \cite{rt1} and $\mathcal{O}$ leaves all sets invariant (see Appendix \ref{appen1} for a proof). Given any set of free states $\mathcal{F}$, a measure of resource can be introduced using the relative entropy $D(\rho|\sigma):=\tr(\rho\log \rho)-\tr(\rho\log \sigma) $, which is given by
\be
\mathcal{M}(\rho)\equiv \min_{\sigma\in \mathcal{F}} D(\rho|\sigma)
\ee
Such a measure automatically satisfies the properties \cite{Chitambar_2019} of positivity ($\mathcal{M}(\rho)\geq0$), faithfulness ($\mathcal{M}(\rho)=0\iff\rho\in\mathcal{F}$) and monotonicity ($\mathcal{M}(\rho)\geq\mathcal{M}(\mathcal{E}(\rho))$ whenever $\mathcal{E}$ is a free operation). In our case, with $\mathcal{F}_G$ as free states, the measure reduces to the measure of non-Gaussianity, denoted as $NG(\rho)$. For any non-Gaussian state $\rho$ with covariance matrix $\Lambda$, its closest Gaussian state is denoted by $\Gamma_\rho$, which has the same covariance matrix $\Lambda$. This allows us to write \cite{NG1}
\be\label{x1}
NG(\rho)=S(\Gamma_\rho)-S(\rho)
\ee
where $S(\rho)$ is the Von-Neumann entropy of the state $\rho$.

\section{Broadcasting of non-Gaussianity}
Let us begin by defining the task of broadcasting. 

\begin{definition}
   Non-Gaussianity can be broadcast if there exists a Gaussian operation $\mathcal{E}$ and a non-Gaussian state $\rho_A$ such that the output state $\mathcal{E}(\rho_A)=\sigma_{AB}$ satisfies $\Tr_B(\sigma_{AB})=\rho_A$ and $\Tr_A(\sigma_{AB})$ is non-Gaussian.
\end{definition}
A schematic of the broadcasting task is illustrated in Figure 1. This definition is the most general one possible, and directly extends to any quantum resource theory. The operation can explicitly be written as $\mathcal{E}(\rho)=\tr_C(\hat{U}_{S}(\rho\otimes\gamma_{BC})\hat{U}_{  S}^\dagger)$ where $\hat{U}_{  S}$ is a Gaussian unitary and $BC$ (the environment $E\equiv BC$) is in a Gaussian state $\gamma$.
\begin{figure}[h]
\includegraphics[scale=0.8]{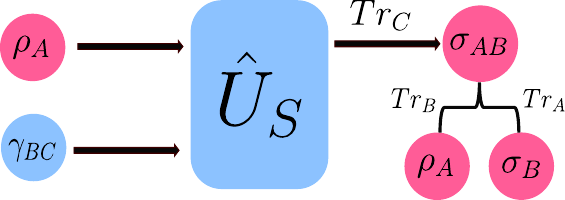}
\caption{\label{fig:epsart} Broadcasting protocol. Everything in blue is Gaussian (including the unitary), and states in pink need to be non-Gaussian for a successful broadcasting protocol. }
\end{figure}
In regard to the task of broadcasting, we have the following proposition:

\begin{proposition}
    Any resource theory that admits a positive, monotonic, faithful, and super-additive measure of resource cannot be broadcast. Here, super-additivity of a measure $\mathcal{M}$ means $\mathcal{M}(\sigma_{AB})\geq \mathcal{M}(\sigma_A)+\mathcal{M}(\sigma_B)$
\end{proposition}
\begin{proof}
    We prove it by contradiction. Suppose there exist a free operation $\mathcal{E}$ and a resourceful state $\rho_A$ for which broadcasting is possible, then it must hold that for $\mathcal{E}(\rho_A)=\sigma_{AB}$ we have $\sigma_A=\rho_A$ and $\mathcal{M}(\sigma_B)>0$ (because $\sigma_B$ must be resourceful and hence by positivity and faithfulness of the measure $\mathcal{M}(\sigma_B)>0$). Then by monotonicity of the measure, $\mathcal{M}(\rho_A)\geq\mathcal{M}(\sigma_{AB})\geq \mathcal{M}(\sigma_A)+\mathcal{M}(\sigma_B)\implies\mathcal{M}(\sigma_B)=0$ which implies $\mathcal{M}(\sigma_B)=0$, a contradiction.  
\end{proof}
From this proposition, if $NG(\rho)$ were super-additive, broadcasting non-Gaussianity would be impossible. However, contrary to common belief, $NG(\rho)$ is not super-additive. Super-additivity of $NG(\rho)$ is equivalent to extremality of Gaussian states for mutual information, as stated in the following lemma.

\begin{lemma}\label{lem1}
    For any non-Gaussian state $\rho_{AB}$ define mutual information $I(\rho_{AB})=S(\rho_A)+S(\rho_B)-S(\rho_{AB})$, and let $\Gamma_{AB}$ denote its closest Gaussian state. Then
    \begin{equation*}
   \{ NG(\rho_{AB})\geq NG(\rho_A)+NG(\rho_B)\}\iff \{I(\rho_{AB})\geq I(\Gamma_{AB})\}
    \end{equation*}
\end{lemma}
{\it Proof sketch.} The result follows from  \eqref{x1} and the fact that if $\Gamma_{AB}$ is the closest Gaussian state to $\rho_{AB}$, then $\Gamma_A$ and $\Gamma_B$ are the closest Gaussian states to $\rho_A$ and $\rho_B$. This holds because $\Gamma_{AB}$ and $\rho_{AB}$ share the same covariance matrix, and so do their reduced states, which are obtained by projecting onto the set of modes of the reduced state required.\\
\\
It has been shown that, however, certain states violate the extremality condition for mutual information \cite{Park_2017}, and hence also violate the super-additivity of $NG(.)$. As a concrete example, we consider two orthogonal cat states $\ket{0}=(\ket{\alpha}+\ket{-\alpha})\sqrt{2}$ and $\ket{1}=(\ket{\alpha}-\ket{-\alpha})/\sqrt{2}$ where $\ket{\alpha}$ is a coherent state ($\hat{a}\ket{\alpha}=\alpha\ket{\alpha}$). These states are normalized and orthogonal for sufficiently large $\alpha$ ($\alpha\gtrapprox2)$. For such ranges of $\alpha$ the orthogonality translates to $\bra{-\alpha}\alpha\rangle\approx 0$. Now for the state $\ket{\phi+}_{AB}=(\ket{00}+\ket{11})/\sqrt{2}$ we explicitly derive in the appendix \ref{appen4} that $\Delta NG_{\phi_+}:= NG(\ket{\phi+}_{AB})-NG(\phi+_{A})-NG(\phi+_{B})$ is given by:
\be\label{eqr1}
\Delta NG_{\phi_+}=2\ln{2}-2h(\sqrt{1+4\alpha^2})+h(\sqrt{1+8\alpha^2})
\ee
where $h(x):=\frac{x+1}{2}\ln{\frac{x+1}{2}}-\frac{x-1}{2}\ln{\frac{x-1}{2}}$ and $\phi+_{A(B)}$ are the reduced states on respective subsystem. As shown in Figure \ref{fig:epsart}, for certain values of $\alpha$, super-additivity of $NG(.)$ is violated.
\begin{figure}[h]
\begin{centering}
\includegraphics[scale=0.5]{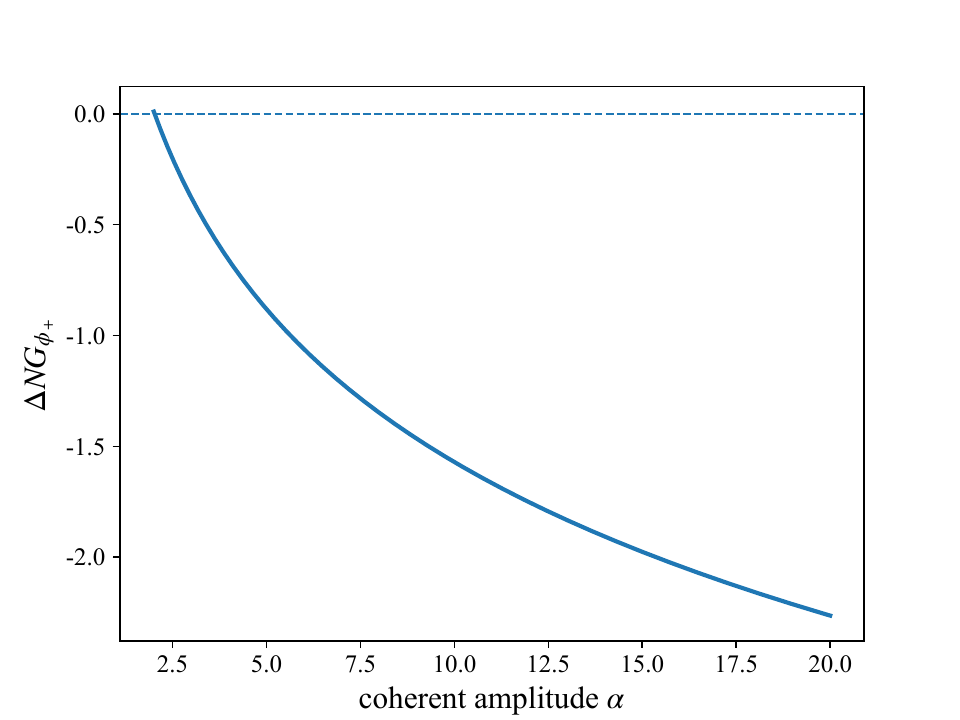}
\caption{\label{fig:epsart} Variation of $\Delta NG_{\phi_+}$ as given by (\ref{eqr1}). The negativity of the values shows that for those $\alpha$ the state violates super-additivity. Refer to the text for more details.}
\end{centering}
\end{figure}
Further, to our knowledge, no known measure of non-Gaussianity satisfies all the properties required to categorically rule out broadcasting. This motivates us to reformulate the problem in a different way.

\subsection{\label{sec:level3} Fixed points of Gaussian channels}
If a broadcasting scenario exists with some Gaussian operation $\mathcal{E}$ and some non-Gaussian state $\rho_A$, then there must also exist a Gaussian operation having $\rho_A$ as its fixed point, namely the operation $\Tr_B\circ \mathcal{E}$. Here, a state is a fixed point of a quantum channel if it remains invariant under the action of the channel. Our first theorem states:

\begin{theorem}\label{t1}
    On the level of the covariance matrix $\Lambda$, the action of any Gaussian channel is given by $\Lambda\to X^T\Lambda X +Y$ (see \eqref{x2}). A Gaussian channel with a Schur stable $X$ cannot admit any non-Gaussian fixed point.
\end{theorem}
The proof is deferred to the Appendix \ref{appen2}. It is important to stress that the Schur stability of $X$ here is the key requirement, which means that all of the eigenvalues $\lambda_i$ of $X$ satisfy $|\lambda_i|<1$. This forces the dynamics in phase space to be strictly contractive, which shifts or perturbs every non-Gaussian state. However, this does not cover all the Gaussian channels. We can define the spectral radius of a matrix M with eigenvalues $\{\lambda_i\}$ as
\be
r(M):=\max_i |\lambda_i|
\ee
The Gaussian channels can now be classified as $r(X)>1,r(X)=1$, or $r(X)<1$. We have already ruled out the existence of non-Gaussian fixed points for $r(X)<1$. Our next theorem rules out $r(X)>1$:

\begin{lemma}\label{tx1}
    A Gaussian channel with $r(X)>1$ cannot admit any fixed point, be it Gaussian or non-Gaussian.
\end{lemma}
{\it Proof sketch.} The essential idea is that if there is a fixed point, the equation
\be
f(M)=X^TMX-M+Y=0,
\ee
must have a solution $M=\Lambda>0$. This equation is the well-known Lyapunov equation \cite{boyd}, important in the analysis of stable system dynamics. Now it is known from standard analysis of such systems \cite{boyd} that having a solution $\Lambda>0$ for a system that has $Y\geq 0$ forces $r(X)\leq 1$ and hence, by negation it follows that $r(X)>1$ admits no fixed point such that $\Lambda>0$ and $Y\geq 0$.   \\
\\
Both the above results suggest that broadcasting cannot take place with Gaussian channels that have $r(X)<1$ or $r(X)>1$ because they fail to admit any non-Gaussian fixed point, which is the first requirement for broadcasting of the state. If we define $\mathcal{S}:=\{X\in \mathbb R^{2n\times 2n}|r(X)<1\}$, then this is a semi-algebraic set as it can be characterized by (via the Schur–Cohn test \cite{Dym_Young_1990}) a finite list of polynomial inequalities in the coefficients of the characteristic polynomial of $X$. Now the boundary of the closure of such sets $\partial \overline{\mathcal S}:=\{X\in \mathbb R^{2n\times 2n}|r(X)=1\}$ is essentially strictly lower dimensional \cite{sinn2014algebraicboundariesconvexsemialgebraic} than $\mathcal{S}$ which makes $\partial \overline{\mathcal S}$ a Lebesgue measure $0$ set in $\mathbb R^{4n^2}$.
This means that channels admitting non-Gaussian fixed points form an exceptional, fine-tuned subset; a randomly chosen Gaussian channel almost surely lies within the contractive region $r(X)<1$ and therefore admits only Gaussian fixed points.
Hence, \textit{if we randomly pick a Gaussian channel, then it will never have a non-Gaussian fixed point and hence, would never broadcast}. Nonetheless, it is very easy to construct channels that lie within this boundary, with a simple example being the two-mode channel $I\otimes \mathcal L_\eta(.)$ where $\mathcal L(.)$ is the loss channel which has $X=\sqrt{\eta}I$ and $Y=(1-\eta)I$. This channel has infinitely many non-Gaussian fixed points given by $\ket{\psi}\otimes\ket{0}$ with $\ket{\psi}$ being any non-Gaussian state. 
\subsection{\label{sec:level3} Broadcasting and degradability of non-Gaussianity}
As far as fixed points are concerned, we have shown that the set of channels that have $X\in\partial \overline{\mathcal S}$ cannot be ruled out. Yet broadcasting is a much more restrictive notion than just having non-Gaussian fixed points. We leverage this now using a key concept called complementary channel. For a channel $\mathcal{E}(\rho_A)=\Tr_E(\hat{U}\rho_A\otimes\omega_E \hat{U}^\dagger)$ and any state $\rho_A$, we define the complementary channel as $\mathcal{E}^c_\rho(\sigma)=\Tr_A(\hat{U}\rho_A\otimes\sigma_E\hat{U}^\dagger)$. This leads to the following theorem.

\begin{theorem}\label{d1}
    If a Gaussian channel $\mathcal{E}$ (with $r(X)=1$) has at least one non-Gaussian fixed point $\rho$, then for that $\rho$, the complementary channel $\mathcal{E}^c_\rho$ is Gaussian.
\end{theorem}
\begin{proof}
    The complementary channel depends on the Gaussian unitary $\hat{U}_S$ defining the Gaussian channel with $S=\begin{pmatrix}
        s_1&s_2\\
        s_3&s_4
    \end{pmatrix}$. If the Gaussian channel has a non-Gaussian fixed point then $\mathcal{E}^m(\rho)=\rho$ which gives the constraint for every $m\geq 1$:
    \be\label{drdo}
    \chi(\xi;\rho)=\chi(X^m(\xi);\rho)\exp({-\frac{1}{4} \xi^T\tilde{Y} \xi})
    \ee
    where $\tilde{Y}=\sum_{j=0}^{m-1}X^{jT}YX^j$. This constraint follows from the fact that the characteristic function must satisfy $\chi(\xi;\rho)=\chi(\xi;\mathcal{E}^m(\rho))$. In the appendix \ref{appen2}, we discuss this constraint and why this means that for any $\xi$, $\xi^T\tilde{Y\xi}$ needs to be finite. Next, we want to prove that $\chi \left({  \xi}_E;\mathcal{E}^c_\rho(\omega)\right)$ is Gaussian whenever $\omega$ is Gaussian. Observe that,
\ba
\chi \left({  \xi};\mathcal{E}^c_\rho(\omega)\right)&=&\chi \left(0,{  \xi};\hat{U}_S\rho_A\otimes\omega_E\hat{U}_S^\dagger\right)\\
&=&\chi \left(s_2\xi,s_4\xi;\rho_A\otimes\omega_E\right)\\
&=&\chi\left(s_2\xi;\rho \right)\chi \left(s_4\xi;\omega_E\right)
\ea
In the first step, we have used property (3) of the Wigner characteristic function, in the second step, we used \eqref{t3} to remove the unitary action, and in the third step, we used property (2) of the Wigner characteristic function. Since $\omega$ is Gaussian, $\chi \left(s_4{  \xi};\omega_E\right)$ is a Gaussian function in $\xi$. To complete the proof we also need $\chi \left(s_2{  \xi};\rho_A\right)$ be gaussian. Let us analyse the space where $s_2\xi$ leaves, which is the image of the operator $s_2$, denoted by $im(s_2)$. Now $im(s_2)=im(Y)$ because $Y=s_2\Lambda_\gamma s_2^T$ with $\Lambda_\gamma>0$ being the covariance matrix of the environment that realizes the channel $\mathcal E$. Now we use a useful result related to the stability of discrete-time systems and control theory, whose proof we defer to the appendix \ref{appen3}
\\
\begin{lemma}\label{betax}
    Given $r(X)=1$ and $E_{<1}$ be the generalized eigenspace associated with eigenvalues of $X$ that lie strictly within the unit disk, and $E_{=1}$ be the generalized eigenspace with eigenvalues on the boundary of the unit disk. Further, define $a_j(v):=||\sqrt{Y}X^j(v)||^2$ and $G(v)=\sum_{j=0}^{\infty}a_j(v)$ then the following are equivalent:
    \begin{itemize}
        \item $G(v)<\infty$ for all $v\in im(Y)$
        \item $im(Y)\subseteq E_{<1}$
    \end{itemize}
    This in turn means $X^j(v)\to 0$ for all $v\in im(Y)$ as $j\to \infty$.
\end{lemma}
The lemma above, when applied to $v=s_2\xi$, means $X^j(v)\to 0$ as $j\to\infty$ while $G(v)$ remains finite. This, along with the characteristic function constraint (\ref{drdo}), gives:
\be
\chi(s_2\xi;\rho)=\exp({-\frac{1}{4} \xi^Ts_2^T\tilde{Y} s_2\xi})=\exp(-\frac{1}{4}G(s_2\xi))
\ee
hence, $\chi(s_2\xi;\rho)$ is also Gaussian which completes the proof.
\end{proof}
We now state the main theorem of the paper.

\begin{theorem}\label{thrm2}
    Non-Gaussianity cannot be broadcast
\end{theorem}
\begin{proof}
    Suppose broadcasting non-Gaussianity with a Gaussian channel $\mathcal{E}$ were possible. Then for some non-Gaussian state $\rho_A$ we would have $\mathcal{E}(\rho_A)=\tr_C(\hat{U}_{S}(\rho_A\otimes\gamma_{BC})\hat{U}_{  S}^\dagger)=\sigma_{AB}$ such that the Gaussian channel $\mathcal{L}=\tr_B\circ\mathcal{E}$ has a fixed point $\rho_A$, $\sigma_B\notin \mathcal{F}_G$ and $\gamma_{BC}$ is Gaussian. For this channel we would have a pair of $(X,Y)$ that defines action of this Gaussian channel and a value $r(X)$. Now theorem \ref{t1} proves that if $r(X)<1$ we fail to have a channel that can admit such a fixed point. Similarly, lemma \ref{tx1} says we fail to admit non-Gaussian fixed point if $r(X)>1$. Now for $r(X)=1$ we can admit such a fixed point but then by theorem \ref{d1}, the complementary channel $\mathcal{L}^c_\rho$ is Gaussian which maps to the space $BC$, because tracing out is Gaussian then this implies that $\tr_C\circ\mathcal{L}^c_\rho$ is also Gaussian. Hence, $\sigma_B=\tr_C\circ\mathcal{L}^c_\rho(\gamma_{BC})$ is a Gaussian state, yet broadcasting requires $\sigma_B$ to be non-Gaussian; this explicit conflict establishes the contradiction.
\end{proof}
More generally, following the ideas of irreversibility and degradation of asymmetry \cite{Asy1}, we can define the notion of irreversibility and degradation of non-Gaussianity. A state conversion $\rho_A\to\sigma_{A'}$ is reversible if there exists a pair of Gaussian operations $(\mathcal{E},\mathcal{R})$ such that $\mathcal{E}(\rho_A)=\sigma_{A'}$ and $\mathcal{R}(\sigma_{A'})=\rho_A$. Otherwise, we say that the non-Gaussianity of $\rho_A$ is degraded under the conversion. Throughout, we require that conversions occur via Gaussian operations. In light of degradability, we state the following theorem:

\begin{theorem}\label{thrm3}
    Consider two systems $A$ and $B$ jointly prepared in the state $\rho_A\otimes\rho_B$. Let them interact via a Gaussian operation $\Lambda$, and let $\mathcal{G}$ be the local map on $A$ induced by $\Lambda$ for a fixed state on $B$. If $\mathcal{G}$ is non-Gaussian, then for some state of $A$, the induced conversion on $B$ is irreversible (see Figure \ref{figt}).
\end{theorem}
\begin{proof}
The proof follows from theorem-\ref{thrm2}, and proceeds by contradiction. The map $\mathcal{G}(.)$ is defined as $\mathcal{G}(.)=\tr_B(\Lambda(.\otimes\rho_B))$. If $\mathcal{G}$ is non-Gaussian, then by definition there exists some Gaussian state $\rho_A$ such that $\mathcal{G}(\rho_A)$ is non-Gaussian. Assume there exists a recovery channel $\mathcal{R}$, then for any non-Gaussian $\rho_B$, the composition channel $\mathcal{R}\circ\Lambda$ would broadcast the non-Gaussianity of $\rho_B$, producing both $\rho_B$ and $\mathcal{G}(\rho_A)$ as non-Gaussian. This contradicts Theorem \ref{thrm2}. Hence, $\mathcal{R}$ cannot exist, and the conversion induced on $B$ is irreversible.
\end{proof}
\begin{figure}[h]
\includegraphics[scale=0.5]{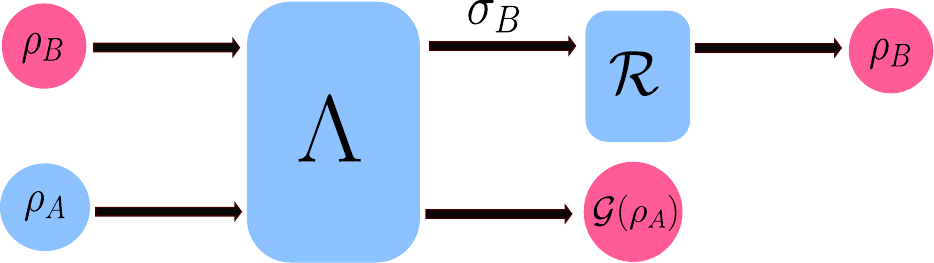}
\caption{\label{figt} The setting for Theorem-\ref{thrm3}. If, using a Gaussian operation $\mathcal{R}$, the state $\rho_B$ can be recovered from $\sigma_B$, then the effective local operation $\mathcal{G}$ is Gaussian and can therefore be implemented without having access to $\rho_B$. The figure shows the contradiction that we can broadcast $\rho_B$ if this is not the case. }
\end{figure}
Finally, from the relation  $\mathcal{F}_G\subsetneq \mathcal{F}^c_G\subsetneq \mathcal{F}_W$ and the fact that the set of Gaussian operations $\mathcal{O}$ preserve all three sets, we can readily conclude that \textit{Gaussian operations cannot broadcast quantum non-Gaussianity as well as Wigner negativity}. If either of these resources could be broadcast, then the associated non-Gaussianity would also be broadcast, contradicting Theorem \ref{thrm2}. However, it remains an open question whether these resources can be broadcast when considering the maximal set of resource non-generating operations $\mathcal{O}_{RNG}$ which is expected to be larger than the Gaussian set. Since $\mathcal{O}_{RNG}$ is not well characterized for either $\mathcal{F}^c_G$ or $ \mathcal{F}_W$, we defer such questions to future work.

\section{Conclusion and outlook}
In summary, we have established that non-Gaussian quantum states cannot be broadcast using Gaussian operations. This no-go theorem extends the family of quantum no-broadcasting results to the continuous-variable domain and reveals an intrinsic limit on the manipulation of non-classical states: one cannot duplicate or distribute the non-Gaussian resource content of a state without invoking non-Gaussian dynamics. Our findings underscore that Gaussian operations – despite their ease of implementation – are fundamentally limited in their ability to propagate non-classical continuous-variable features. In essence, if one mode gains non-Gaussianity through a Gaussian interaction, another mode must lose an equivalent amount.

Our results carry implications for quantum information processing. For example, in continuous-variable quantum computing, one might hope to take a highly non-Gaussian state (such as a GKP sensor state \cite{gkp} with negative Wigner function) and share its non-Gaussian advantages among multiple computing modules using only Gaussian channels (beam splitters, squeezers, etc.). Our no-broadcasting theorem shows that such strategies will fail: at most one module can end up with a non-Gaussian state if the process is Gaussian. Any multi-party distribution of non-Gaussian resources must involve non-Gaussian operations at some stage. 

On the theoretical side, our approach demonstrates how techniques from Gaussian channel theory and control theory (Lyapunov stability analysis) can yield powerful constraints on quantum resource dynamics. It would be interesting to explore whether partial or approximate broadcasting of non-Gaussianity might be possible under certain relaxed conditions (for instance, allowing a small injection of non-Gaussian noise, or aiming only to broadcast a limited degree of non-Gaussian character). Another promising direction is to quantify the trade-off of non-Gaussianity between subsystems more precisely, perhaps by developing new monotones or operational measures of non-Gaussian resource flow. We hope that our work stimulates further investigation into the interplay between Gaussian processes and non-Gaussian resources, and helps inform the design of continuous-variable quantum technologies where non-Gaussian states serve as key ingredients.

\begin{acknowledgments}
K.C. thanks Niklas Budinger and Tanmoy Pandit for useful discussions. We acknowledge support from the Danish National Research Foundation (bigQ, DNRF0142) and EU ERC (ClusterQ, grant agreement no. 101055224, ERC-2021-ADG).
\end{acknowledgments}

\appendix

\section{Properties of characteristic function and Wigner function}\label{appen1}
Although the equations \ref{chip2} and \ref{chip3} are well known in the CV community yet we give small self-contained proofs for them.
\begin{proof}
    The main ingredient in the proof is that for the Weyl displacement operator we have:
    \be
    \begin{split}
    \hat{D}(\xi_A,\xi_B)&=\exp\left({i(\hat{x}_A, \hat{x}_B)(\xi_A,\xi_B)}\right)\\
&=\exp\left({i(\hat{x}_A\xi_A+\hat{x}_B\xi_B)}\right)\\
&= \hat{D}(\xi_A)\otimes\hat{D}(\xi_B)
\end{split}
    \ee
Now using this it is easy to see that $\chi\left({  \xi}_A,{  \xi}_B;\rho_A\otimes\rho_B\right)=\tr(\rho_A\otimes\rho_B \hat{D}(\xi_A)\otimes\hat{D}(\xi_B))=
\chi\left({  \xi}_A;\rho_A\right)\chi\left({  \xi}_B;\rho_B\right)$. Similarly, $\chi\left(0,{  \xi}_B;\rho_{AB}\right)=\tr(\rho_{AB}\calI\otimes\hat{D}(\xi_B))=\tr_B(\rho_B\hat{D}(\xi_B))=\chi(\xi_B;\rho_B)$
\end{proof}
The set of Gaussian operations maps Gaussian states to themselves. Now for any state $\rho\in \mathcal{F}^c_G$ and any Gaussian operation $\mathcal{G}$ we have $\mathcal{G}(\rho)=\sum_i p_i \mathcal{G}(\rho_i)$ where each $\rho_i$ is a Gaussian state and hence $\mathcal{G}(\rho)\in \mathcal{F}^c_G$ which shows that the set of Gaussian operations ($\mathcal{O}$) keeps $\mathcal{F}^c_G$ invariant. However, the argument that $\mathcal{O}$ keeps $\mathcal{F}_W$ invariant is not straightforward, and we could not find a direct proof in the literature; hence, we give it below:
\begin{proof}
    Given that we know how the characteristic function transforms under general Gaussian operations (equation-\ref{chip1x}), we find the transformation for the Wigner function:
    \be
    \begin{split}
        W(x;\mathcal{E}(\rho))\propto \int d^{2n}r\,e^{-i r^{T} x}
 e^{-\tfrac14 r^{T}Y\,r}
 \chi(X r;\rho). 
    \end{split}
    \ee
Now we can substitute $\chi(X r;\rho)=
\int d^{2n}y\,
 e^{i r^{T}X^T y}\,
 W(y;\rho)$. This gives
\be\label{app1}
W(x;\mathcal E(\rho))=
\int d^{2n}y\;W(y;\rho)\,G_{X,Y}(x\!\mid\!y),
\ee
with the kernel $G_{X,Y}(x\!\mid\!y)\propto
\int d^{2n}r\,
\exp\!\bigl[-\tfrac14 r^{T}Y\,r
 -i r^{T}\,(X^Ty-x)\bigr]$. This kernel for $Y=0$ collapses to $G_{X,0}(x\!\mid\!y)=\delta(X^Ty-x)$ and hence, $W(x;\mathcal E(\rho))=W(X^{-T}x;\rho)$ which is positive whenever $W(x;\rho)$ is positive. For general $Y\geq 0$ if we define $v=X^Ty-x$ then we have: 
\be
\begin{split}
G_{X,Y}(x\!\mid\!y)&\propto e^{-v^TY^{-1}v} \delta(\Pi_{ker}(Y)(v))
\end{split}
\ee
Now this is a Gaussian probability density kernel modulated with a Dirac delta function and hence always positive, and so we see that by equation-\ref{app1}, $W(x;\mathcal E(\rho))$ is always positive whenever $W(y;\rho)$ is positive and $\mathcal{E}$ is a general Gaussian operation.
\end{proof}
\section{Proof of Theorem 1}\label{appen2}
\begin{proof}
     Schur stability of a matrix $X$ means that the spectral radius of the matrix is less than one. For any Schur stable matrix and any real vector $\Vec{r}$, we have
    \be\label{sch}
    \lim_{n\to\infty} X^n(\Vec{r})=0
    \ee
Now, a single application of the channel maps $V\to X^TVX+Y$, and using mathematical induction, it is straightforward to prove that $m$ successive applications of the channel maps
\be
V\to X^{Tm} V X^{m}+\sum_{n=0}^{m-1}X^{Tn} Y X^{n}
\ee
This means that the characteristic function changes as
\be
\chi(\xi;\mathcal E^{\otimes m}(\rho))=\exp(-\frac{1}{4} \xi^T\tilde{Y} \xi)\chi(X^m(\xi);\rho)
\ee
where $\tilde{Y}=\sum_{n=0}^{m-1}X^{Tn} Y X^{n}\geq 0$. Now taking $m\to\infty$ and using the equation \ref{sch} we see
\be
\begin{split}
\chi(\xi;\mathcal E^{\otimes m}(\rho))&=\exp(-\frac{1}{4} \xi^T\tilde{Y} \xi)\chi(0;\rho)\\
&=\exp(-\frac{1}{4} \xi^T\tilde{Y} \xi)\\
&=\chi(\xi;\rho)
\end{split}
\ee
where the second last equality is the normalization condition and the last equality means that $\rho$ is a fixed point for the channel $\mathcal{E}$. But now we see that any admissible fixed point must have a Gaussian characteristic function which completes the proof 
\end{proof}
The condition of having a fixed point 
\be
\chi(\xi;\rho)=\exp(-\frac{1}{4} \xi^T\tilde{Y} \xi)\chi(X^m(\xi);\rho)
\ee
for all $m\geq 0$ also means that for every $\xi\in \mathbb R^{2n}$ $\xi^T\tilde{Y}\xi<\infty$ (in regards to convergence of the sum in $\tilde{Y}$). This is because if there exists a $\xi$ for which this diverges then for $t\xi$ also it will diverge ($t\in\mathbb R$) and hence $\chi(t\xi; \rho)=0$ for any $t$ but by taking $t\to 0$ we would wrongly conclude that $\chi(0;\rho)=0$ because $\chi(\xi;\rho)$ needs to be continuous. \\
\section{Analysis of $r(X)=1$ case}\label{appen3}
Here we give more detailed proof and analysis of the lemma \ref{betax} that we use in our theorem \ref{d1}. What we want to prove is that demanding the sum $G(v)$ (as defined in the main text) to converge for all $v\in im(Y)$ necessarily implies $im(Y)\subseteq E_{<1}$. The other implication $im(Y)\subseteq E_{<1}\implies G(v)<\infty, \forall v\in im(Y) $ is more obvious because over $E_{<1}$ space $X$ acts like a Schur stable matrix \cite{kai,boyd}. 
\begin{proof}
We begin by assuming a contradiction $\exists v\in (im(Y)/\{0\})\bigcap E_{=1}$ such that $v=\sum_{l}v_l$ and the sum $G(v)$ converges. Here $l\in \{+1,-1,\theta\}$ which labels three types of subspaces within $E_{=1}$ that is generalised eigenspaces of $+1,-1$ and of phases that occur in pair $e^{\pm i\theta}$. We first deal with the $l=\pm 1$ blocks over which $X$ will have a Jordan block form of $(\pm I+N)$ with $N$ being nilpotent such that $N^{r+1}=0$. In this block, we have
\be
a_j(v_{l})=||\sum_{k=0}^{r}\binom{j}{k} c_k(l)||^2
\ee
with $c_k(l)=\sqrt{Y}\tilde{u}_k(l)$ with $\tilde{u}_k(l)=(l)^kN^k(v)$. Now, due to the growth of binomial coefficients, one shows that
\be
\lim_{j\to\infty}\frac{\sum_{k<r}\binom{j}{k}||c_k(l)||}{\binom{j}{r}||c_r(l)||}=0
\ee
This means there exists $j\geq J$ such that $\sum_{k<r}\binom{j}{k}||c_k(l)||\leq\frac{1}{2}\binom{j}{r}||c_r(l)||$ with which we bound
\be
\begin{split}
    \sqrt{a_j(v_l)}\geq \binom{j}{r}||c_r(l)||-\sum_{k<r}\binom{j}{k}||c_k(l)||\geq \frac{1}{2}\binom{j}{r}||c_r(l)||
\end{split}
\ee
and this bound proves the divergence of $G(v_l)=\sum_j a_j(v_l)$. Now we look at the rotation blocks where $X=R_\theta=\begin{pmatrix}
    cos(\theta)&-sin(\theta)\\
    sin(\theta)& cos(\theta)
\end{pmatrix}$ and $a_j(v_l,\theta)=v_l^TR_{-j\theta}YR_{j\theta}v_l$. observe that $j=0$ gives $a_0=v_l^TYv_l>0$ because $v_l\notin ker(Y)$. If $\theta=\frac{2\pi q}{p},(q,p)\in\mathbb{N}^2$ then, the sequence $a_j(v_l,\theta)$ is periodic with $a_j(v_l,\theta)=a_{j+p}(v_l,\theta)$ which gives $\sum_{j=0}^{Np-1}a_j\geq Na_0$ which leads to divergence of $G(v_l)$ when $N\to\infty$. For $\theta$ not a rational multiple of $2\pi$ we can approximate the sum with an integral as
\be
\frac{1}{N}\sum_{j=0}^N a_j(v_l,\theta) \to \frac{1}{2\pi}\int d\theta a_j(v_l,\theta)=\frac{\tr(Y)||v_l||^2}{2}
\ee
for large $N$ and we see that this again gives $G(v)$ diverging as linear in $N$. The above arguments show that whenever $v$ has support over any of $E_{=1}$ subspaces, then the quantity $G(v_l)$ diverges, which leads to divergence of $G(v)$, which is in contradiction to what we assumed and hence, $im(v)\subseteq E_{=1}$. 
\end{proof}
The sum $G(v)$ also goes by the name of Gramian and plays an important role in the analysis of detectability and controllability of systems. The detectability of unobservable modes of such systems is possible iff these modes are asymptotically stable \cite{kai}, which means they exist on spaces where $X$ is Schur stable.  
\section{Analysis of super-additivity for entangled cat states}\label{appen4}
As discussed in the main text we are interested in computing $\Delta NG_{\phi_+}$ which by Lemma~\ref{lem1} is same as computing the difference of mutual information $I(\ket{\phi+})-I(\Gamma_{\phi+})$ where $\ket{\phi+}_{AB}=\frac{1}{\sqrt{2}}(\ket{\alpha}\ket{\alpha}+\ket{-\alpha}\ket{-\alpha})$ and $\Gamma_{\phi+}$ is the Gaussian reference state (the closest Gaussian state). Now $I(\ket{\phi+})=2\ln{2}$ because it is maximally entangled. $I(\Gamma_{\phi+})$ can be calculated using the covariance matrix of of $\Gamma_{\phi+}$ which is same as that of $\ket{\phi+}$ and is given by
\be
\Lambda_{\phi+}=\begin{pmatrix} 4\alpha^2+1 & 0& 4\alpha^2 & 0\\
0&1& 0& 4\alpha^2\\
4\alpha^2&0& 4\alpha^2+1& 0\\
0&4\alpha^2&0& 1\end{pmatrix}
\ee
From this we also get the covariance matrix of reduced state as $diag(4\alpha^2+1,1)$ which gives the entropy of $S(\Gamma_{\phi+A})+S(\Gamma_{\phi+B})=2h(\sqrt{4\alpha^2+1})$. Here we have used the well known expression of entropy of Gaussian states as given in \cite{Braunstein_2005}. Now from the global covariance matrix we can find the symplectic eigenvalues which are given by $(\sqrt{1+8\alpha^2},1)$ which gives $S(\Gamma_{\phi+})=h(\sqrt{1+8\alpha^2})$.Overall mutual information $I(\Gamma_{\phi+})=S(\Gamma_{\phi+A})+S(\Gamma_{\phi+B})-S(\Gamma_{\phi+})=2h(\sqrt{4\alpha^2+1})-h(\sqrt{8\alpha^2+1})$.

\bibliography{apssamp}

\end{document}